\documentclass[journal,a4paper]{IEEEtran}
\usepackage{amsthm}
\usepackage{cite}
\usepackage[pdftex]{graphicx}
\ifCLASSINFOpdf
\else
\fi
\usepackage{amsmath}
\usepackage{amssymb}
\usepackage{mathtools}

\DeclareMathOperator*{\argmax}{arg\,max}
\usepackage{algorithm}
\usepackage{algorithmic}
\makeatletter
\def\BState{\State\hskip-\ALG@thistlm}
\makeatother
\usepackage{fixltx2e}
\usepackage{stfloats}

\usepackage{afterpage}
\usepackage{lipsum}
\usepackage[margin=0.75in]{geometry}
\newgeometry{top=115pt, left=54pt, right=37pt, bottom=54pt}
\usepackage{enumitem}

\begin{document}
\setlength{\topsep}{0pt}
\setlength{\belowcaptionskip}{0pt}
\setlength{\abovedisplayskip}{0.9pt}
\setlength{\belowdisplayskip}{0.9pt}
\setlength{\abovedisplayshortskip}{0.8pt}
\setlength{\belowdisplayshortskip}{0.8pt}
\setlength{\parskip}{0.7pt}
\setlength{\textfloatsep}{0.9pt}
\setlength{\floatsep}{.7pt}
%
% paper title
% Titles are generally capitalized except for words such as a, an, and, as,
% at, but, by, for, in, nor, of, on, or, the, to and up, which are usually
% not capitalized unless they are the first or last word of the title.
% Linebreaks \\ can be used within to get better formatting as desired.
% Do not put math or special symbols in the title.
\title{Online Charge Scheduling for Electric Vehicles in Autonomous Mobility on Demand Fleets}

% author names and affiliations
% use a multiple column layout for up to three different
% affiliations
\author{\IEEEauthorblockN{Nathaniel Tucker \quad}
%\IEEEauthorblockA{School of Electrical and\\Computer Engineering\\
%University of California Santa Barbara\\
%Santa Barbara, California\\
%Email: nathaniel\_tucker@umail.ucsb.edu}
\and
\IEEEauthorblockN{Berkay Turan \quad}
%\IEEEauthorblockA{School of Electrical and\\Computer Engineering\\
%University of California Santa Barbara\\
%Santa Barbara, California\\
%Email: }}
\and
\IEEEauthorblockN{Mahnoosh Alizadeh}
%\IEEEauthorblockA{School of Electrical and\\Computer Engineering\\
%University of California Santa Barbara\\
%Santa Barbara, California\\
%Email: }}
\vspace*{-0.45cm}}

% conference papers do not typically use \thanks and this command
% is locked out in conference mode. If really needed, such as for
% the acknowledgment of grants, issue a \IEEEoverridecommandlockouts
% after \documentclass

% for over three affiliations, or if they all won't fit within the width
% of the page, use this alternative format:
% 
%\author{\IEEEauthorblockN{Michael Shell\IEEEauthorrefmark{1},
%Homer Simpson\IEEEauthorrefmark{2},
%James Kirk\IEEEauthorrefmark{3}, 
%Montgomery Scott\IEEEauthorrefmark{3} and
%Eldon Tyrell\IEEEauthorrefmark{4}}
%\IEEEauthorblockA{\IEEEauthorrefmark{1}School of Electrical and Computer Engineering\\
%Georgia Institute of Technology,
%Atlanta, Georgia 30332--0250\\ Email: see http://www.michaelshell.org/contact.html}
%\IEEEauthorblockA{\IEEEauthorrefmark{2}Twentieth Century Fox, Springfield, USA\\
%Email: homer@thesimpsons.com}
%\IEEEauthorblockA{\IEEEauthorrefmark{3}Starfleet Academy, San Francisco, California 96678-2391\\
%Telephone: (800) 555--1212, Fax: (888) 555--1212}
%\IEEEauthorblockA{\IEEEauthorrefmark{4}Tyrell Inc., 123 Replicant Street, Los Angeles, California 90210--4321}}

% use for special paper notices
%\IEEEspecialpapernotice{(Invited Paper)}
%\afterpage{\newgeometry{top=0.75in, left=0.75in,right=0.75in,bottom=0.75in}}

% make the title area
\maketitle

% As a general rule, do not put math, special symbols or citations
% in the abstract
\begin{abstract}
% Operating a fleet of autonomous-mobility-on-demand (AMoD) electric vehicles (EVs) is no  simple task. Specifically, scheduling when and where the EVs should charge and routing the EVs to different regions for rebalancing becomes troublesome when considering multiple regions, limited charging resources, operational costs, and stochastic renewable generation. Furthermore, due to the unknown daily sequences of ride requests, the problem cannot be solved by any \textit{offline} approach. 
In this paper, we study an online charge scheduling strategy for fleets of autonomous-mobility-on-demand electric vechicles (AMoD EVs). We consider the case where  vehicles complete   trips and then enter a between-ride state  throughout the day, with their information becoming available to the fleet operator in an online fashion. In the between-ride state, the vehicles must be scheduled for charging and then routed to their next passenger pick-up locations. Additionally, due to the unknown daily sequences of ride requests, the problem cannot be solved by any offline approach. As such, we study an online welfare maximization heuristic based on primal-dual methods that allocates limited fleet charging resources and rebalances the vehicles while avoiding congestion at charging facilities and pick-up locations. We discuss a competitive ratio result comparing the performance of our online solution to the clairvoyant offline solution and provide numerical results highlighting the performance of our heuristic.
\end{abstract}

% \begin{IEEEkeywords}
% Electric vehicles, Autonomous-mobility-on-demand, resource allocation, primal-dual, online algorithm.
% \end{IEEEkeywords}
%no keywords

% For peer review papers, you can put extra information on the cover
% page as needed:
% \ifCLASSOPTIONpeerreview
% \begin{center} \bfseries EDICS Category: 3-BBND \end{center}
% \fi
%
% For peerreview papers, this IEEEtran command inserts a page break and
% creates the second title. It will be ignored for other modes.
\IEEEpeerreviewmaketitle
\newtheorem{theorem}{Theorem}
\newtheorem{lemma}{Lemma}
\newtheorem{proposition}{Proposition}
\newtheorem{corollary}{Corollary}[lemma]
\makeatletter
\def\blfootnote{\xdef\@thefnmark{}\@footnotetext}
\makeatother

\blfootnote{\footnotesize%This paragraph of the first footnote will contain the date on which you submitted your paper for review. It will also contain support information, including sponsor and financial support acknowledgment. \\
\hspace{-8.5pt}This work was supported by  the NSF Grants 1837125, 1847096.\\
N. Tucker, B. Turan and M. Alizadeh are with the Department of Electrical and Computer Engineering, University of California, Santa Barbara, CA 93106 USA (email: nathaniel\_tucker@ucsb.edu).}
\section{Introduction}
\label{section: intro}
Three developing technologies in the transportation sector have the potential to revolutionize the paradigm of \textit{personal urban mobility}: autonomous vehicles (self-driving or driverless vehicles), mobility-on-demand (car-sharing or ride-sharing), and plug-in electric vehicles \cite{Pavone_book, ACC_Models}. These technologies have independently garnered much research and experimentation; however, literature addressing the potential synergies is still emerging \cite{Greenblatt2015}. Consequently, we consider the welfare maximization problem for a fleet dispatcher who operates a large number of Autonomous-Mobility-on-Demand Electric Vehicles (AMoD EVs). Because of the real-time requirements of AMoD systems, we propose a novel online solution for optimizing the charging and rebalancing processes of a fleet of AMoD EVs.

Regarding AMoD fleets, much work has been done focusing on matching riders with vehicles, routing vehicles to destinations, and rebalancing the vehicles throughout a set of pick-up/drop-off locations \cite{between_ride_review, between_ride_matching,between_ride_parallel,between_ride_on_demand,between_ride_decomposition}. There is also work in the area of coordinated charging for fleets of AMoD EVs. Paper \cite{xiaoqi} gives an overview of managing AMoD fleets and energy services in future smart cities. The authors of \cite{zhang_rossi_pavone} utilize a model predictive approach to optimize charge scheduling and routing in an AMoD system. Paper \cite{chen_kockelman_hanna} presents a study of the operations of a AMoD fleet including the implications of vehicle and charging infrastructure decisions. Furthermore, \cite{chen_kockelman} studies the implications of pricing schemes on an AMoD fleet. Work has also been done in congestion aware \cite{salazar2019congestion} and model predictive routing methods \cite{salazar_mpc} for AMoD systems. Additionally, \cite{rossi_iglesias_alizadeh,pavone_transit} study interactions between AMoD systems with the power grid and public transit. %A common element missing in many of these previous works is an online control strategy.

Regarding charging strategies for large populations of EVs, papers \cite{New_TSG_SmartEVGrid,New_Survey,New_20} provide in-depth reviews and studies of smart charging technologies. %Furthermore, \cite{New_24,New_23,New_22} provide coordination methods for EV charging for the purposes of valley filling, shaving peak load, and minimizing load fluctuation. 
An important but less studied issue is that the benefits of smart charging can be severely limited if usage of the shared Electric Vehicle Supply Equipment (EVSEs) is uncontrolled \cite{bryce}. Moreover, without access control and allocation strategies within charging facilities, EVSEs can become congested while other EVSEs are left empty. This limits the smart charging benefits as congested EVSEs are forced to charge one EV after another to satisfy demand. To address this issue, papers including \cite{Robu,auc2charge,differentiated,menu} have studied smart charging, admission control, and resource allocation for facilities equipped with EVSEs.

In this manuscript, we aim to complement both the recent work in smart charging and AMoD fleet routing with the objective of optimizing AMoD fleet charging and rebalancing processes in an online fashion. Specifically, we study an online heuristic that schedules fleet charging, allocates limited fleet resources, and rebalances the vehicles while avoiding congestion at charging facilities and pick-up locations. Moreover, our methodology does not rely on statistics of the daily sessions (unlike model predictive approaches). The work presented in this manuscript complements existing literature on smart charging and fleet management and the main contributions are as follows: 
\begin{itemize}
    \item The online heuristic makes decisions for multiple charging facilities, multiple pick-up/drop-off regions, constrained fleet charging resources, operational costs, renewable generation integration, and rebalancing.
    \item The online heuristic does not rely on fractional allocations or rounding methods to produce integer assignments in a computationally feasible manner.
    \item The online heuristic readily handles the inherent stochasticity of the fleet scheduling problem without requiring statistics on the future inputs of the system.
    \item The online heuristic accounts for worst case (i.e., adversarially chosen between-ride sequences) and always yields welfare within a factor of $\frac{1}{\alpha}$ of the optimal offline. 
\end{itemize}
\textit{Organization:} Section \ref{section: system} describes the AMoD EV fleet charge scheduling problem and system model. Section \ref{section: Online_Mech} presents the online heuristic that updates the dual variables when solving the online problem and discusses the online heuristic's properties and performance guarantees. Section \ref{section:numerical} discusses numerical results showing the performance of our heuristic.
%%%%%%%%%%%%%%%%%%%%%%%%%%%%%%%%%%%%%%%%%%%%%%%

\newgeometry{top=105pt, left=54pt, right=37pt, bottom=54pt}

%%%%%%%%%%%%%%%%%%%%%%%%%%%%%%%%%%%%%%%%%%%%%%%%
\section{Preliminaries}
\label{section: system}
\subsection{Fleet Objective}
\label{section:problem description}
In this section, we describe the charge scheduling problem for a fleet of AMoD EVs that are servicing customers within a city. The objective of the fleet dispatcher is to maximize profit by optimizing the between-ride schedules of the EVs to exploit cheaper time-of-use electricity rates and behind-the-meter solar generation in addition to efficiently distributing the vehicles throughout the area. In this manuscript, we use capitalized calligraphic variables (e.g., $\mathcal{Z}$) to represent indexed sets (e.g., $z = 1,\dots,Z$).

We consider an area of operation consisting of a set of regions $\mathcal{D} = \{ 1,\dots,D\}$ where each region $d\in\mathcal{D}$ can be viewed as a destination for a vehicle to pick up its next customer. %Each region $d$ has an expected minimum number of customers requesting rides at time $t$ that we denote as $\lambda_d(t)$. Additionally, 
Each region $d$ has a maximum capacity $\Omega_d(t)$ that limits the number of AMoD vehicles in the area (e.g., due to municipal restrictions, congestion mitigation, or ride demand forecasts). Within this area of operation, we consider $J$ consecutive between-ride sessions occurring within a time span $t=1,\dots, T$. Each between-ride session $j\in\mathcal{J}$ begins at time $t_j^-$ when a vehicle in the fleet completes a previous ride. Information for the $j$th between-ride session could be revealed earlier when the EV picks up the passenger rather than at drop-off; however, due to unknown traffic conditions and travel times, the session information is not available until $t_j^-$ (drop-off). %We note that there is not a one-to-one mapping from fleet vehicles to between-ride sessions; therefore, some EVs can have multiple between-ride sessions within the time span and other EVs might never enter the between-ride state (i.e., if an EV is servicing a customer with a far destination, it might not enter state 1 within the given time span). 
At this time, the fleet dispatcher must decide what the vehicle should do next. Namely, a schedule must be created for the between-ride session that includes a charging facility, a desired charge amount, and the next customer pickup destination (if the EV has sufficient battery level, the schedule may skip charging altogether). A simple example for three vehicles' between-ride schedules is portrayed in Fig. \ref{fig: example}.
\begin{figure}
    \centering
    \includegraphics[width=1.0\columnwidth]{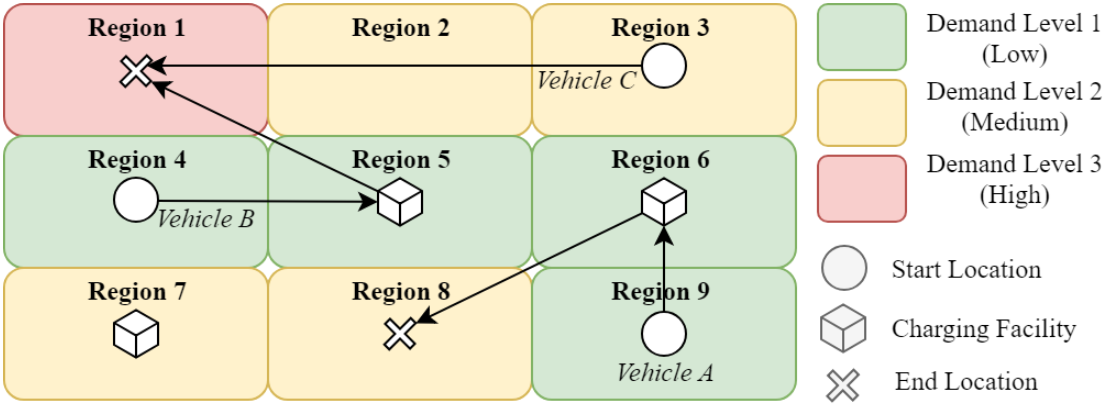}
    \caption{Example between-ride schedules for 3 vehicles.}%\vspace{16}}
    \label{fig: example}
\end{figure}
\subsection{Charging Model}
\label{section:charging model}
There is a set of regions $\mathcal{F} \subseteq \mathcal{D}$ that have charging facilities with Electric-Vehicle-Supply-Equipment (EVSEs) where the EVs can replenish their batteries. Each charging facility $f\in\mathcal{F}$ is equipped with $M_f$ Single-Output-Multiple-Cable (SOMC) EVSEs each with $C_f$ output cables meaning each facility can fit up to $M_f C_f$ EVs simultaneously (but only charge $M_f$ simultaneously) \cite{SOMC, ntucker_allerton,ntucker_ACC,ntucker_TSG}. Additionally, due to power delivery limitations of the EVSE hardware, each EVSE can deliver up to $E_f$ units of energy in one time slot. Each facility can procure energy from two sources: 1) an on-site solar generation system and 2) the local power distribution grid. The available solar energy at facility $f$ at time $t$ is denoted as $\delta_f(t) \in [0,\Delta_f ]$ where $\Delta_f$ is the maximum power rating of the on-site system at facility $f$. Energy can also be purchased from the local distribution grid at price $\pi_f(t)$ per unit. Due to the local transformer's operational limits, facility $f$ is constrained to purchase less than $\mu_f(t)$ units of energy from the grid in each time slot.
\subsection{Operational Cost Model}
\label{section:operational cost model}
%Throughout the entire time period, the fleet dispatcher aims to maximize earnings; therefore, the fleet dispatcher wants to limit the amount of time the vehicles are in the between-ride state as they are not yielding profit. Moreover, 
Since we are solely focused on optimizing a given number of between-ride sessions and are not explicitly modeling the rest of the AMoD system operation problem, we need to penalize vehicles for remaining out-of-service (i.e., charging or traveling without customers) for extended periods of time. This is also essential to compare the quality of our online solution to that of the clairvoyant offline solution. 
%Let us assume, the fleet dispatcher utilizes a large fleet size to facilitate the circulation of reserve vehicles while others are out-of-service. 
As such, we assume the fleet dispatcher incurs a virtual cost, $\phi(t)$, per out-of-service vehicle at each time $t$. Moreover, $\phi(t)$ is permitted to be time-varying to give the fleet dispatcher control over the quantity of active fleet vehicles throughout a given day. A large $\phi(t)$ could be used to ensure that less vehicles are out-of-service during peak demand periods.

\subsection{Offline and Online Problem}
\label{section:offline and online}
In this work, we first formulate the charge scheduling problem as an offline optimization and then relate the offline problem to the online problem. In the offline case, we assume the fleet dispatcher is clairvoyant and knows the entire sequence of $J$ events over the time span $t=1,\dots,T$. As such, the offline fleet dispatcher can create the optimal schedules for the between-ride sessions and can achieve maximal profit. However, the reality is that the dispatcher does not know the customers' desired destinations, pick-up to drop-off travel times, traffic conditions, or electricity grid conditions in advance. Instead, the between-ride sessions are revealed one-by-one throughout the time span meaning that an online solution method is required for real world implementation. Moreover, the charge scheduling problem presents challenges not easily overcome in many online solutions; namely, the lack of accurate statistics on the between-ride sessions as there are many exogenous factors that directly affect when and where between-ride sessions begin. These factors include ride demand, road congestion, weather, popular events, and construction delays, to name a few. As such, we present an online solution that can account for adversarially chosen sequences of between-ride sessions and still yield profit that is within a constant factor of the clairvoyant offline solution.

\subsection{Between-Ride Schedules}
% Each schedule $s$ for between-ride session $j$ has value $v_{js}$ for the fleet dispatcher calculated as:
% \begin{align}
%     \label{eqn:valuation}
%     v_{js} = V_d( SoC^+_{js} ) + v_{d_{js}}.
% \end{align}
% Here, $SoC^+_{js}$ represents the \textit{State of Charge} of the EV when it reaches its next pickup location. The function $V_d(\cdot)$ calculates the fleet dispatcher's valuation of the final energy level of the EV in schedule $s$ at location $d$ (maintaining sufficient energy levels in the EVs is critical to provide uninterrupted rides to customers). We note that the fleet dispatcher is free to choose $V_d(\cdot)$ for their desired operational objectives. The variable $v_{d_{js}}$ represents the profit that the fleet dispatcher receives for picking up a customer at destination $d_{js}$. The destinations present different values to the fleet dispatcher due to exogenous factors such as ride demand, location, weather, etc. 
In this section, we describe the parameters associated with each between-ride schedule.
Each between-ride session $j\in\mathcal{J}$ begins at time $t_j^-$ when an AMoD EV drops off a passenger. At this time, a set of feasible between-ride schedules is generated based on the vehicle's current battery level and location. We denote the set of feasible between-ride schedules for session $j$ as $\mathcal{S}_j$. Each schedule $s\in\mathcal{S}_j$ consists of the following components: 1) $t_j^-$: time when all between-ride schedules start for session $j$; 2) $d_j^-$: start location of all schedules; 3) $t_{js}^+$: end time for schedule $s$; 4) $d^+_{js}(t)$: binary (0,1) indicator function for when the EV reaches its end destination $d_{js}^+$ in schedule $s$ of session $j$; %$d^+_{js}(t) = \delta(t-t_{js}^+)$. 
%We note that $d_{js}^+(t_{js}^+) = 1$ and  $d_{js}^+(t) = 0, \;\forall t \neq t_{js}^+$, 
5) $o_{js}(t)$: out-of-service indicator that is set to 1 for $t\in [t_j^-, t_{js}^+ ]$ and 0 otherwise; 6) $v_{js}$: value of schedule $s$ to the fleet dispatcher given by \begin{align}
    \label{eqn:valuation}
    v_{js} = V_d( SoC^+_{js} ) + v_{d_{js}}.
\end{align}
Here, $SoC^+_{js}$ represents the \textit{State of Charge} of the EV when it reaches its next pickup location. The function $V_d(\cdot)$ calculates the fleet dispatcher's valuation of the final energy level of the EV in schedule $s$ at location $d$ (maintaining sufficient energy levels in the EVs is critical to provide uninterrupted rides to customers). We note that the fleet dispatcher is free to choose $V_d(\cdot)$ for their desired operational objectives. The variable $v_{d_{js}}$ represents the profit that the fleet dispatcher receives for picking up a customer at destination $d_{js}$. The destinations present different values to the fleet dispatcher due to exogenous factors such as ride demand, location, weather, etc; 7) $c^{mf}_{js}(t)$: binary (0,1) cable reservation variable that is set to 1 if the EV in session $j$ will use a charging cable from EVSE $m$ at facility $f$ during time step $t$ in schedule $s$; 8) $e_{js}^{mf}(t)$: energy delivered to EV at time $t$ in session $j$ at EVSE $m$ at facility $f$ in schedule $s$. Through $e_{js}^{mf}(t)$, the fleet dispatcher is able to customize when the EV will receive charge, when it will be idle, and the rate of charge (from a discrete set). This effectively allows the fleet dispatcher to \textit{smart charge} the EVs (e.g., exploit cheaper time-of-use electricity rates or behind-the-meter solar). The feasible schedules for between-ride session $j$ can be written as:
\begin{equation}
    \{ t_{j}^-, \{o_{js}(t)\}, \{ c_{js}^{mf}(t) \}, \{ e_{js}^{mf}(t) \}, d_j^-, \{d_{js}^+(t)\}, t_{js}^+, \{v_{js}\} \}.
\end{equation}
\noindent Furthermore, the fleet dispatcher sets the variable $x_{js}$ to 1 if schedule $s$ is selected and 0 otherwise. If no schedule is desirable, the dispatcher routes the vehicle to the central depot where the vehicle will wait to be assigned later on.
\subsection{Cost Model and Offline Scheduling Problem}
In order to facilitate charge scheduling and vehicle routing, the fleet dispatcher maintains a total count for each shared resource across all assigned schedules. The variables $y_c^{mf}(t)$ and $y_e^{mf}(t)$ correspond to the total allocated cables and energy, respectively, at facility $f$ at EVSE $m$ at time $t$. Additionally, each facility has to procure the energy needed by all the EVSEs within. The total energy procurement at facility $f$ at time $t$ is denoted as $y_g^{f}(t)$. The resource demands at charging facilities are calculated in equations \eqref{eqn:offline cable demand}-\eqref{eqn:offline gen demand}:
\begin{align}
    \label{eqn:offline cable demand}
    & y_{c}^{mf}(t)=\sum_{\mathcal{J},\mathcal{S}_j} c_{js}^{mf}(t) x_{js},\\
    \label{eqn:offline energy demand}
    & y_{e}^{mf}(t)=\sum_{\mathcal{J},\mathcal{S}_j} e_{js}^{mf}(t) x_{js},\\
    \label{eqn:offline gen demand}
    & y_{g}^{f}(t)=\sum_{\mathcal{J},\mathcal{S}_j,\mathcal{M}_f} e_{js}^{mf}(t) x_{js}.
\end{align}
Similarly, the fleet dispatcher counts the vehicles in the between-ride state and the number of vehicles committed to destinations in the variables $y_o(t)$ and $y_d(t)$, respectively:
\begin{align}
    \label{eqn:out of service count}
    & y_{o}(t)=\sum_{\mathcal{J},\mathcal{S}_j} o_{js}(t) x_{js},\\
    \label{eqn:offline dest demand}
    & y_{d}(t)=\sum_{\mathcal{J},\mathcal{S}_j} d_{js}^+(t) x_{js}.
\end{align}
Due to the number of vehicles allocated to each resource, the fleet dispatcher incurs various costs to serve the between-ride schedules. The energy procurement, $y_g^{f}(t)$, determines the generation cost of facility $f$:
\begin{align}
\label{eq:generation cost}
    &G_{f}(y_{g}^{f}(t)) = \\
    &\nonumber\begin{cases}
        0 & y_g^{f}(t) \in [0,\delta_f(t)]\\
        \pi_f(t) (y_{g}^{f}(t)-\delta_f(t)) & y_g^{f}(t) \in \big(\delta_f(t),\delta_f(t)+\mu_f(t)] \\
        +\infty & y_{g}^{f}(t) > \delta_f(t)+\mu_f(t).
        \end{cases}
\end{align}
% The number of vehicles committed to each pickup destination, $y_{d}(t)$, leads to the following cost:
% \begin{align}
% \label{eq:reroute cost}
%     &R_{d}(y_{d}(t)) = \\
%     &\nonumber\begin{cases}
%         0 & y_{d}(t) \in [0,\lambda_d(t)\big]\\
%         \omega_d(t) (y_{d}(t)-\lambda_d(t)) & y_{d}(t) \in \big(\lambda_d(t),\Omega_d(t)]\\
%         +\infty & y_{d}(t) > \Omega_d(t).
%         \end{cases}
% \end{align}
Namely, electricity is free while solar is available, else the facility purchases from the grid at a rate of $\pi_f(t)$ per unit. 

Additionally, vehicles that are charging and traveling to their next pickup destination are not able to serve customers. As described in Section \ref{section:operational cost model}, the fleet dispatcher incurs a penalty proportional to the number of out-of-service vehicles:
%The out-of-service vehicles, $y_{o}(t)$, can lead to profit losses proportional to the  expected  value  per  serviced ride:
\begin{align}
\label{eq:outofservice cost}
    &O(y_{o}(t)) = 
    \begin{cases}
        \phi(t) y_{o}(t), & y_o(t) \leq I(t)\\
        +\infty & y_o(t) > I(t),
    \end{cases}
\end{align}
where $I(t)$ is the maximum number of out-of-service vehicles that the fleet dispatcher allows at time $t$. If the AMoD fleet dispatcher has full knowledge of all between-ride sessions $j\in\mathcal{J}$ over the entire time span $t=1,\dots,T$, the primal offline optimization is as follows:
\begin{subequations}
\begin{align}
    \label{eqn:offline obj}
    &\max_{x} \sum_{\mathcal{J}, \mathcal{S}_j} v_{js} x_{js}-\sum_{\mathcal{T}, \mathcal{F}}    G_{f}(y_{g}^{f}(t)) %-\sum_{\mathcal{T}, \mathcal{D}}    R_{d}(y_{d}(t)) 
    -\sum_{\mathcal{T}}    O(y_{o}(t)) \\* \nonumber
    % & \hspace{20pt}-\sum_{\mathcal{T}, \mathcal{L}, \mathcal{M}_l} f_{e}^{ml}(y_{e}^{ml}(t)) \\\nonumber
    % & \hspace{20}-\sum_{\mathcal{T}, \mathcal{L}}    f_{g}^{l}(y_{g}^{l}(t)) \\\nonumber
    &\textrm{ subject to:}\nonumber\\*
    \label{eqn:offline one option}
    &\hspace{20pt}\sum_{\mathcal{S}_j} x_{js} \leq 1, \hspace{17pt}\forall\; j\\*
    &\hspace{20pt}\label{eqn:offline integer}\hspace{1pt} x_{js}  \in \{ 0,1 \}, \hspace{15pt}\forall \; j, s\\*
    %&\label{eqn:offline cable lim}\hspace{1pt}y_{c}^{ml}(t) \leq C_l, \hspace{13pt} \forall\; l,m,t\\
    &\hspace{20pt}\label{eqn:offline cable lim}\hspace{1pt}y_{c}^{mf}(t) \leq C_f, \hspace{11pt} \forall\; f,m,t\\*
    &\hspace{20pt}\label{eqn:offline energy lim}\hspace{1pt}y_{e}^{mf}(t) \leq E_f, \hspace{11pt} \forall\; f,m,t\\*
    &\hspace{20pt}\label{eqn:offline dest cap} \hspace{1pt}y_{d}(t) \leq \Omega_d(t), \hspace{8pt} \forall\; d,t\\*
    &\hspace{20pt}\hspace{1pt} \nonumber \textrm{and }\eqref{eqn:offline cable demand}, \eqref{eqn:offline energy demand}, \eqref{eqn:offline gen demand}, \eqref{eqn:out of service count},\eqref{eqn:offline dest demand}.
    %&    \label{eq: demand geq 0}y_{r}^{ml}(t)  \geq 0, \\*
    %&\nonumber\hspace{49pt}\forall r\in\{c,e,g\}, m\in\mathcal{M}_{l},     l\in\mathcal{L}, t\in[0,T].
\end{align}
\end{subequations}
%\noindent In \eqref{eqn:offline obj}, the objective is to maximize the total value from the between-ride sessions minus charging facility generation costs and costs due to out-of-service vehicles. 
In \eqref{eqn:offline obj}, the objective maximizes fleet dispatcher's utility by distributing the AMoD EVs throughout all regions $d\in\mathcal{D}$ while minimizing the operational costs due to charging facilities and out-of-service vehicles. Specifically, the last term of \eqref{eqn:offline obj} limits the duration of the between-ride sessions to increase earnings and decrease the need to use extra vehicles. Constraint \eqref{eqn:offline one option} ensures only one schedule is chosen per between-ride session, \eqref{eqn:offline integer} is an integral constraint on the decision variable, \eqref{eqn:offline cable lim} ensures the cable limit is not exceeded at each EVSE, \eqref{eqn:offline energy lim} ensures the EVSE energy limits are not exceeded, and \eqref{eqn:offline dest cap} enforces the vehicle limit in each region. By temporarily relaxing the integral constraint \eqref{eqn:offline integer}, the problem can be examined in the dual domain (however, we note that our competitive ratio results are with respect to integer allocations). Specifically, we make use of Fenchel duality with dual variables $u_j$, $p_c^{mf}(t)$, $p_e^{mf}(t)$, $p_g^{f}(t)$, $p_d(t)$, and $p_o(t)$ \cite{Fenchel}. In the following, the Fenchel conjugate of a function $f(y(t))$ is given as:
\begin{equation}
    f^*(p(t)) = \sup_{y(t)\geq0} \big\{ p(t)y(t) - f(y(t)) \big\}.
\end{equation} 

\noindent
As such, the offline Fenchel dual of \eqref{eqn:offline obj}-\eqref{eqn:offline dest cap} is as follows:
\begin{subequations}
\begin{alignat}{3}
    \label{eqn:dual obj}
    \min_{u,p} &\sum_{\mathcal{J}} u_j 
    + \sum_{\mathcal{T}, \mathcal{D}} 
    R_{d}^*(p_{d}(t))\\*
    \nonumber&+ \sum_{\mathcal{T}, \mathcal{F}}    G_{f}^*(p_{g}^{f}(t))
    + \sum_{\mathcal{T}} O^*(p_{o}(t))\\*
    \nonumber&+ \hspace{-8pt}\sum_{\mathcal{T}, \mathcal{F},\mathcal{M}_f} \Big( K_{c}^{mf*}(p_{c}^{mf}(t)) + K_{e}^{mf*}(p_{e}^{mf}(t)) \Big)
    %&\nonumber+ \sum_{\mathcal{T}, \mathcal{D}} 
    %R_{d}^*(p_{p}^{d}(t))\\*
    %& \nonumber\hspace{51pt}+ \sum_{\mathcal{T}, \mathcal{L},\mathcal{M}_l} f_{e}^{ml}^*(p_{e}^{ml}(t)) \\*
\end{alignat}
subject to:
%\nonumber &\textrm{subject to:}&&\\*
\begin{alignat}{2}
    \label{eqn:dual user utility}
    &u_j \geq \;v_{js}&& 
    -  d_{js}^+(t_{js}^+)p_d(t_{js}^+) - \sum_{\mathcal{T}} \Big( c_{js}^{mf}(t)p_{c}^{mf}(t)\\*
    & &&+ e_{js}^{mf}(t)\big[p_{e}^{mf}(t)+p_{g}^{f}(t)\big] + o_{js}(t) p_o(t) \Big)\nonumber\\*
    & &&\hspace{3pt}\forall\;j,s,f,m,d\nonumber\\*
    \label{eqn:dual ui}
    &u_j \geq \;0, &&\hspace{3pt}\forall \;j\\*
    \label{eqn:dual pj}
    & p_{c}^{mf}(t),\;p&&_{e}^{mf}(t),\;p_{g}^{f}(t),\;p_{d}(t),\;p_{o}(t) \geq \;0, \\*
    & &&\hspace{3pt}\nonumber\forall\; f,m,d,t,
\end{alignat}
\end{subequations}
where $R_d^*(p_d(t))$, $G_f^*(p_g^f(t))$, $O^*(p_o(t))$, $K_c^{mf*}(p_c^{mf}(t))$, and $K_e^{mf*}(p_e^{mf}(t))$ are the Fenchel conjugates for the regional vehicle limit, the facility generation cost, the out-of-service cost, the EVSE cable constraint, and the EVSE energy constraint, respectively. The Fenchel conjugates for the cable and energy constraints, respectively, are as follows:
\begin{align}
\label{eq:fenchel cable cost}
    &K^{mf*}_{c}(p^{mf}_{c}(t)) = p^{mf}_{c}(t)C_f, \hspace{24pt} p^{mf}_{c}(t) \geq 0,\\
    \label{eq:fenchel energy cost}
    &K^{mf*}_{e}(p^{mf}_{e}(t)) = p^{mf}_{e}(t)E_f, \hspace{24pt} p^{mf}_{e}(t) \geq 0.
\end{align}
The Fenchel conjugate for the energy procurement cost function at each facility can be written as:
\begin{align}
\label{eq:fenchel gen cost}
    &G_f^{*}(p^{f}_{g}(t)) =  \\*
    &\nonumber\begin{cases}
        \delta_f(t)p^{f}_{g}(t), &p^{f}_{g}(t) < \pi_f(t) \\
        (\delta_f(t)+\mu_f(t))p^{f}_g(t)-\mu_f(t)\pi_f(t) & p^{f}_{g}(t) \geq \pi_f(t). \\
    \end{cases}
\end{align}
\\
The Fenchel conjugate for the regional vehicle limit is:
\begin{align}
\label{eq:fenchel reroute cost}
    &R_d^{*}(p_{d}(t)) = p_d(t)\Omega_d(t), \hspace{20pt} p_d(t) \geq 0.
\end{align}
% The Fenchel conjugate for the operational cost for overallocated vehicles can be written as follows:
% \begin{align}
% \label{eq:fenchel reroute cost}
%     &R_d^{*}(p_{d}(t)) =  \\*
%     &\nonumber\begin{cases}
%         \lambda_d(t)p_{d}(t), &p_{d}(t) < \omega_d(t) \\
%         \Omega_d(t) p_{d}(t)-\omega_d(t)\big(\Omega_d(t)-\lambda_d(t)\big), & p_{d}(t) \geq \omega_d(t). \\
%     \end{cases}
% \end{align}
Lastly, the Fenchel conjugate for the penalty for the out-of-service vehicles is as follows:
\begin{align}
\label{eq:fenchel outofservice}
    &O^{*}(p_{o}(t)) = 
    \begin{cases}
        0, &p_{o}(t) < \phi(t) \\
        (p_{o}(t) - \phi(t))I(t), & p_{o}(t) \geq \phi(t). \\
    \end{cases}
\end{align}
%\noindent where $I$ is the total number of AMoD EVs in the fleet. %We assume $I$ is sufficiently large to account for customer demand and the out-of-service vehicles.
\subsection{Scheduling Decisions}
In the offline case, let us examine the Fenchel dual \eqref{eqn:dual obj}-\eqref{eqn:dual pj} when choosing charging schedules for the between-ride sessions. Specifically, examining the KKT conditions for constraint \eqref{eqn:dual user utility} reveals the optimal scheduling decisions. If a between-ride session yields $u_j \leq 0$, then the AMoD EV is not needed to serve customers at that time or electricity prices are too high; therefore, $u_j$ is set to 0 and the vehicle is routed to the central depot. Alternatively, when the fleet dispatcher wants vehicles to charge and serve customers, $u_j$ will be positive. If $u_j>0$, the optimal schedule for session $j$ can be found by finding the schedule $s\in\mathcal{S}_j$ that results in the maximal $u_j$: 
\begin{align}
\label{eq: u_n}
    \nonumber u_j =&\max_{s\in\mathcal{S}_j} \Big\{v_{js} 
    -  p_d(t_{js}^+) d^+_{js}(t_{js}^+)
    - \sum_{t\in[t_j^-,t_{js}^+]} \Big( o_{js}(t)p_{o}(t) \\*&
    + c_{js}^{mf}(t)p_{c}^{mf}(t) +e_{js}^{mf}(t)\big[p_{e}^{mf}(t)+p_{g}^{f}(t) \big]\Big)\Big\}.
\end{align}
We note that $u_j$ corresponds to the utility gained from between-ride session $j$ for the fleet dispatcher. Furthermore, the optimization problems in \eqref{eqn:offline obj}-\eqref{eqn:offline dest cap} and \eqref{eqn:dual obj}-\eqref{eqn:dual pj} require full knowledge of the between-ride sessions beforehand. %In practice, due to the unknown destinations and durations of the ride requests, vehicles will enter the between-ride state at different times throughout the time span, prohibiting the clairvoyant offline solution by the fleet dispatcher. 
As discussed in Section \ref{section:offline and online}, the fleet dispatcher must utilize an online solution that can schedule vehicles without any knowledge of the future (i.e., without knowledge of the optimal dual variables). In the following, we discuss such an online heuristic for the between-ride charge scheduling problem.
% \begin{align}
% \label{eq: payment}
%     \hat{p}_{no}^{ml} = \sum_{\mathcal{T}} \Big( c_{no}^{ml}(t)p_{c}^{ml}(t)+e_{no}^{ml}(t)(p_{e}^{ml}(t)+p_{g}^{l}(t)) \Big).
% \end{align}
% \\
\section{Online Scheduling heuristic}
\label{section: Online_Mech}
\subsection{Online Scheduling}

Because the between-ride sessions are revealed throughout the time span, it is apparent that an online solution method is required. 
% In many online problems, approximate dynamic programming (ADP) techniques that make use of statistics have been popular even in systems with large state-spaces. %However, performance guarantees can be difficult to find for ADP heuristics. 
% However, the between-ride charge scheduling problem presents challenges not easily overcome in ADP heuristics; namely, the lack of accurate statistics on the between-ride sessions as there are many exogenous factors that directly affect when and where between-ride sessions begin. These factors include ride demand, road congestion, weather, popular events, and construction delays, to name a few. Additionally, energy prices and solar generation amounts can vary significantly. As such, we present an online solution that can account for any adversarial chosen sequence of between-ride sessions and still yield performance guarantees.
We consider a heuristic that updates the dual variables in an online fashion as between-ride sessions are revealed and then solves equation \eqref{eq: u_n} for each session. The online scheduling heuristic updates the dual variables for each resource based only on the amount of resource that has been allocated up to the current time (i.e., only utilizing the resource allocation counts). The online heuristic serves two main purposes: 1) it ensures that each between-ride schedule yields more value to the fleet dispatcher than the operational cost pertaining to the schedule, and 2) if the demand for rides is low enough or electricity prices are high enough such that no schedule nets positive utility, vehicles are sent back to the central depot. This eliminates further costs from vehicles circulating without serving riders. The underlying framework for the heuristic we use is akin to that of \cite{IaaS}, where the authors present an auction for allocating computing resource bundles at data centers for the purpose of cloud computing and virtual machine provisioning.

In our online scheduling heuristic, we make use of specialized functions proposed in \cite{IaaS} that approximate the optimal dual variables throughout the time span. These dual variable update functions increase slowly at first, but then increase rapidly as the number of allocated resources approach the capacity limits. Furthermore, when the amount of allocated resource is at the capacity limit, the update functions output dual variables high enough such that no schedule will yield positive utility in \eqref{eq: u_n}, thus enforcing the hard capacity limits. The updating function for the dual variable associated with the SOMC EVSE cables at charging facilities is written as follows:
\begin{align}
\label{eqn:zero inf price}
    p_{c}^{mf}(y_{c}^{mf}(t)) =& \Big(\frac{L_c}{2\Psi}\Big)
    \Big( \frac{2\Psi U_c}{L_c} \Big)^{\frac{y_{c}^{mf}(t)}{C_f}},
\end{align}
where $\Psi$ is the total number of shared resources within the fleet system:
\begin{equation}
    \Psi=2\sum_{\mathcal{F}}M_f+D+F+1.
\end{equation}
Furthermore, $L_c$ and $U_c$ correspond to the minimum and maximum value per cable per time unit, respectively. The online scheduling heuristic requires knowledge of $L_c$ and $U_c$ beforehand to set the initial values and to ensure capacity limits are not breached:
\begin{subequations}
\begin{align}
    & \label{eqn:Lc}L_c = \min_{\mathcal{J},\mathcal{S}_j,\mathcal{F},\mathcal{M}_f} \frac{v_{js}}{\Psi\sum_{t\in[t_j^-,t_{js}^+]}c_{js}^{mf}(t)},\\
    &\label{eqn:Uc}U_c = \max_{\mathcal{J},\mathcal{S}_j, \mathcal{F}, \mathcal{M}_f,\mathcal{T}}\; \frac{v_{js}}{c_{js}^{mf}(t)}, \quad c_{js}^{mf}(t)\neq 0.
\end{align}
\end{subequations}
The EVSE charging power, facility generation, out-of-service cost, and destination vehicle limit require similar lower and upper bounds on valuations: $L_e$, $U_e$, $L_g$, $U_g$, $L_o$, $U_o$, $L_d$, and $U_d$, respectively. These are calculated as in equations \eqref{eqn:Lc} and \eqref{eqn:Uc} with the corresponding variables to replace $c_{js}^{mf}(t)$. 

The dual variable update function for the dual variable associated with the SOMC EVSE energy limitations at charging facilities is as follows:
\begin{align}
\label{eqn:zero inf price energy}
    p_{e}^{mf}(y_{e}^{mf}(t)) =& \Big(\frac{L_e}{2\Psi}\Big)
    \Big( \frac{2\Psi U_e}{L_e} \Big)^{\frac{y_{e}^{mf}(t)}{E_f}}.
\end{align}
The update function of the dual variable for the piecewise linear generation cost at each charging facility is more complex than \eqref{eqn:zero inf price} and \eqref{eqn:zero inf price energy}. It has to account for the free solar generation as well as the linear price to procure energy from the local distribution grid. As such, we propose the following dual variable update function:
\begin{align}
\label{eqn:newpricing}
    &p_g^f(y_g^f(t)) = \\*
    &\nonumber\begin{cases}
        \Big( \frac{L_g}{2\Psi} \Big) \Big( \frac{2\Psi\pi_f(t)}{L_g} \Big)^{ \frac{y_g^f(t)}{\delta_f(t)}}, \hspace{87pt} y_g^f(t) < \delta_f(t), \\
        \Big( \frac{L_g-\pi_f(t)}{2\Psi} \Big) \Big( \frac{2\Psi(U_g-\pi_f(t))}{L_g-\pi_f(t)} \Big)^{ \frac{y_g^f(t)}{\delta_f(t)+\mu_f(t)}}+\pi_f(t), \\
        \hspace{175pt} y_g^f(t) \geq \delta_f(t). \\
        \end{cases}
    %&\nonumber\textrm{where } R=\sum_{\mathcal{L}}M_l(C_l+E_l+\frac{1}{M_l}).
\end{align}
The heuristic dual variable update function for the vehicle limits at region $d$ can be written as follows:
\begin{align}
\label{eqn:newpricing_reroute}
    p_{d}(y_{d}(t)) =& \Big(\frac{L_d}{2\Psi}\Big)
    \Big( \frac{2\Psi U_d}{L_d} \Big)^{\frac{y_{d}(t)}{\Omega_d(t)}}.
\end{align}
% The heuristic dual variable update function for the operational cost for overallocated vehicles at each destination $d$ is also a piecewise function and can be written as follows:
% \begin{align}
% \label{eqn:newpricing_reroute}
%     &p_d(y_d(t)) = \\*
%     &\nonumber\begin{cases}
%         \Big( \frac{L_d}{2\Psi} \Big) \Big( \frac{2\Psi\omega_d(t)}{L_d} \Big)^{ \frac{y_d(t)}{\lambda_d(t)}}, \hspace{87pt} y_d(t) < \lambda_d(t), \\
%         \Big( \frac{L_d-\omega_d(t)}{2\Psi} \Big) \Big( \frac{2\Psi(U_d-\omega_d(t))}{L_d-\omega_d(t)} \Big)^{ \frac{y_d(t)}{\Omega(t)}}+\omega_d(t), \hspace{6pt} y_d(t) \geq \lambda_d(t). \\
%         \end{cases}
%     %&\nonumber\textrm{where } R=\sum_{\mathcal{L}}M_l(C_l+E_l+\frac{1}{M_l}).
% \end{align}
Last, the penalty from vehicles in the out-of-service state also requires a heuristic update function for the dual variable $p_o(y_o(t))$ which can be written as:
\begin{align}
\label{eqn:newpricing_outofservice}
    &p_o(y_o(t)) =
        \Big( \frac{L_o-\phi(t)}{2\Psi} \Big) \Big( \frac{2\Psi(U_o-\phi(t))}{L_o-\phi(t)} \Big)^{ \frac{y_o(t)}{I(t)}}+\phi(t).
    %&\nonumber\textrm{where } R=\sum_{\mathcal{L}}M_l(C_l+E_l+\frac{1}{M_l}).
\end{align}
With the 5 aforementioned dual variable update functions \eqref{eqn:zero inf price}, \eqref{eqn:zero inf price energy}, \eqref{eqn:newpricing}, \eqref{eqn:newpricing_reroute}, and \eqref{eqn:newpricing_outofservice}, we now have surrogate functions to use in place of the optimal dual variables in order to solve equation \eqref{eq: u_n} in an online fashion (i.e., at the inception of each between-ride session).

\subsection{Procedure and Performance Guarantees}
The step-by-step procedure for scheduling between-ride sessions for a fleet of AMoD EVs is outlined in Algorithm \textsc{onlineAMoDscheduling}. Namely, at the inception of between-ride session $j$, the fleet dispatcher generates feasible schedules $\mathcal{S}_j$ and then the best schedule, $s^*$, is chosen in line \ref{s_star} which makes use of the heuristically updated dual variables. After every between-ride schedule selection, the fleet dispatcher updates the dual variables accounting for the total amounts of allocated resources.
\begin{algorithm}[]
\small
    \caption{\textsc{onlineAMoDscheduling}}
    \label{algorithm}
    \begin{algorithmic}
    \STATE \textbf{Input:} $I, J, \mathcal{F}, \mathcal{M}_f,C_f, E_f, \mu_f, \delta_f, \pi_f,$
    \STATE \hspace{30pt}$\mathcal{D}, \Omega_d, \phi, \Psi, \{L,U\}_{c,e,g,d,o}$
    \STATE \textbf{Output:} $x, p$
    \end{algorithmic}
    \begin{algorithmic}[1]
    \STATE Define $G_f(\cdot)$ and $O(\cdot)$ according to \eqref{eq:generation cost} and \eqref{eq:outofservice cost}.
    \STATE Define and initialize the dual update functions \eqref{eqn:zero inf price}, \eqref{eqn:zero inf price energy}, \eqref{eqn:newpricing}, \eqref{eqn:newpricing_reroute}, and \eqref{eqn:newpricing_outofservice}.
    \STATE Initialize $x_{js}=0$, $y(t)=0$, $u_j=0$.
    \STATE \textbf{At the inception of between-ride session $j$:}
    \STATE Generate feasible charging/pickup schedules $\mathcal{S}_j$.
    \STATE Update dual variable $u_j$ according to \eqref{eq: u_n}. \label{alg:utility}
    \IF{$u_j > 0$}
        \STATE \label{s_star}$s^{\star} =\argmax_{\mathcal{S}_j}\big\{v_{js} -  p_d(t_{js}^+) d^+_{js}(t_{js}^+)$ \\
        \vspace{3pt}
        \hspace{35pt} $- \sum_{t\in[t_j^-,t_{js}^+]} \big[ o_{js}(t)p_{o}(t)+c_{js}^{mf}(t)p_{c}^{mf}(t)$\\
        \vspace{3pt}
        \hspace{35pt} $ +e_{js}^{mf}(t)[p_{e}^{mf}(t)+p_{g}^{f}(t)] \big]\big\}$
        \vspace{3pt}
        %\STATE $\hat{p}_{no^{\star}}^{m^{\star}l^{\star}} = \sum_{t\in[t_n^-,t_n^+]} \Big( c_{no^{\star}}^{m^{\star}l^{\star}}(t)p_{c}^{m^{\star}l^{\star}}(t)$\\*
        %\hspace{20pt} $+e_{no^{\star}}^{m^{\star}l^{\star}}(t)(p_{e}^{m^{\star}l^{\%star}}(t) +p_{g}^{l^{\star}}(t)) \Big)$
        \vspace{3pt}
        \STATE $x_{js^{\star}}=1$ and $x_{js}=0$ for all $s \neq s^{\star}$
        \STATE Update demand $y(t)$ for cables, energy, generation, destination, and out-of-service according to \eqref{eqn:offline cable demand}-\eqref{eqn:offline dest demand}. \label{alg:demand update}
        \STATE Update dual variables $p(t)$ according to \eqref{eqn:zero inf price}, \eqref{eqn:zero inf price energy}, \eqref{eqn:newpricing}, \eqref{eqn:newpricing_reroute}, and \eqref{eqn:newpricing_outofservice}. \label{alg:price update}
    \ELSE
        \STATE Send the AMoD EV to central depot to re-enter system later and set $x_{js}=0$, \hspace{3pt}  $\forall$ $s \in \mathcal{S}_j$.
    % \ENDIF 
    % \IF{$\exists s^{\star}$ and $x_{js^{\star}}=1$}
    %     \STATE Schedule vehicle $i$ and allocate cables and energy in facility $f^{\star}$ at EVSE $m^{\star}$ and increment the vehicle count at destination $d^{\star}$ by 1.
    %     %\STATE Charge user $n$ at $\hat{p}_{no^{\star}}^{m^{\star}l^{\star}}$.
    % \ELSE
    \ENDIF
    \end{algorithmic}
\end{algorithm}

We can compare the total welfare generated from our online solution to that of the clairvoyant offline solution in the form of a competitive ratio. An online heuristic is described as $\alpha$-competitive when the ratio of welfare from the clairvoyant offline solution to the welfare from the online heuristic is bounded by $\alpha\geq1$. For the between-ride charge scheduling problem, we extend a previous competitive ratio result from \cite{IaaS}. In this work, we assume that each between-ride session utilizes only a small amount of the available resources, thus ensuring that the allocation of one schedule does not adversely affect too many future sessions. %Additionally, we restrict the number of between-ride sessions to be less than the number of vehicles in the fleet. %Restricting $J < I$ ensures that both the offline and online solutions will create the same number of schedules for the same between-ride sessions. %For the sake of brevity of this manuscript, the proofs for the lemmas and theorem are in the appendix of the online version of this paper \cite{ITSC_ARXIV}.
\begin{theorem}
The online heuristic \textsc{onlineAMoDscheduling} in Algorithm \ref{algorithm} is $\alpha$-competitive in welfare across all fleet resources for the fleet dispatcher over $J$ between-ride sessions where $\alpha = \max \big\{ \alpha_1, \alpha_2, \alpha_3, \alpha_4, \alpha_5 \big\}$.
\end{theorem}
\begin{proof}
From Lemmas \ref{lemma cables}-\ref{lemma outofservice} (in Appendix), we have welfare guarantees $\alpha_1,\dots,\alpha_5$ for each of the shared resources. To find the $\alpha\geq1$ for the entire system, we take the maximum over $\alpha_1,\dots,\alpha_5$ to yield the bound that accounts for all resources.
\end{proof}
\section{Numerical Results}
\label{section:numerical}
% \begin{figure}
%     \centering
%     \includegraphics[width=1.00\columnwidth]{SanJose.PNG}
%     \caption{San Jose, CA ride demand split into $D=46$ regions. Ride demand high to low: Red, Yellow, Green. Purple Outline: Charging facility available. Gray Regions: No demand.}
%     \label{fig:SanJose}
% \end{figure}
In this section, we discuss numerical results showing the performance of our online heuristic. In the following simulation, electricity prices and solar generation data were sourced from actual California ISO data in the Bay Area from 2018 \cite{LMPweb},\cite{Solar}. We simulated for a fleet operating in San Jose, CA with $D=46$ regions and $F=8$ charging facilities. Each charging facility is identical with $M_f = 10$ and $C_f=4$. Each of the 8 facilities has on-site solar with a maximum generation of 256 kWh. Likewise, each facility can purchase energy from the grid up to 256 kWh per time step. We set the penalty for out-of-service vehicles equal to 2$\times$ the highest grid electricity price to penalize lengthy between-ride durations. Valuations for each of the regions were either \$15, \$10, or \$5 (red, yellow, green, in Fig. \ref{fig:SanJose_example}, respectively) based on daily ride demand \cite{sherpa}. The regions have AMoD vehicle limits set to 40 (separate from facility capacity if there is a facility in the region). Each vehicle entered the system with either 25\%, 50\%, or 75\% battery level (50 kWh batteries) and were allowed make charge requests in increments of 12.5 kWh. The EVSEs at the charging facilities were limited to deliver either 0 kWh per time slot or 5 kWh per time slot. Furthermore, $V(100\%) =\$10, V(75\%) =\$7.5, V(50\%) =\$5,$ and $V(25\%) = \$2.5$. We also included a linear penalty (\$2 per region traveled) to the schedule valuations to devalue long between-ride routes.
\begin{figure}
    \centering
    \includegraphics[width=1.0\columnwidth]{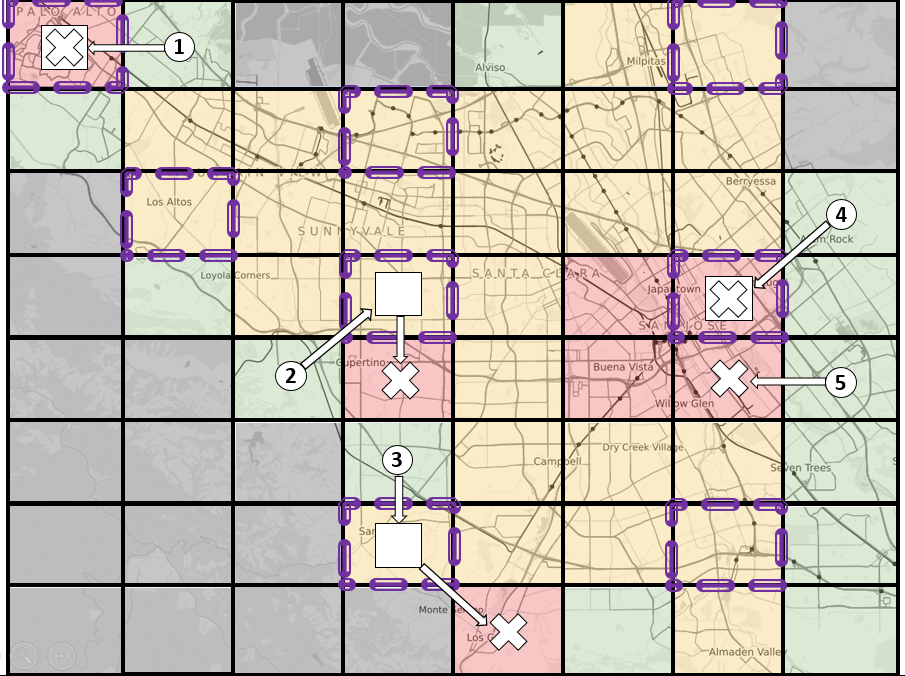}
    \caption{San Jose, CA ride demand split into $D=46$ regions. Ride demand high to low: Red, Yellow, Green. Purple Outline: Charging facility available. Gray Regions: No demand. Also Shown: Most popular between-ride routes for 5 different starting points. Circle: Between-ride session start location. Square: Charging session at facility. Cross: Between-ride end location.}
    \label{fig:SanJose_example}
\end{figure}

In Figure \ref{fig:SanJose_example} we show the most popular between-ride routes for vehicles starting at 5 of the regions (randomly chosen). %We note that in the cases where there is a charging facility at a high value (red) region, the AMoD vehicles tend to charge and end at the same region.
% \begin{figure}
%     \centering
%     \includegraphics[width=1.0\columnwidth]{images/welfaregraph.png}
%     \caption{Fleet dispatcher welfare across 100 days.}
%     \label{fig:welfare}
% \end{figure}
\begin{figure}
    \centering
    \includegraphics[width=1.0\columnwidth]{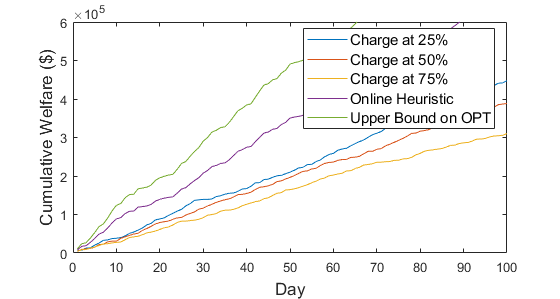}
    \caption{Fleet dispatcher welfare across 100 days.}
    \label{fig:welfare}
\end{figure}
In Figure \ref{fig:welfare}, we compare the welfare generated by our online heuristic to 3 conservative online methods and an upper bound on the optimal solution. Specifically, we compare our heuristic to 3 threshold policies where each EV automatically charges if it is below 75\%, 50\%, or 25\% at the nearest facility and then routes itself to the closest high value destination that is not at the AMoD vehicle limit. As seen in Figure \ref{fig:welfare}, our heuristic is able to consistently outperform these threshold strategies. This is primarily due to shifting charging to time slots when there is available solar or cheaper time-of-use electricity rates (the threshold strategies charge once the vehicle is plugged in and do not stop until fully charged). Because the state-space of any system such as this grows exponentially to intractable sizes, we instead present an upper bound on the optimal offline solution. Accordingly, we calculated the upper bound by relaxing the capacity limits at each facility and destination to allow for all sessions to select their utility maximizing schedule without being constrained by capacity limits.

\section{Conclusion}
In this manuscript, we studied an online heuristic that provides an approximate solution to the AMoD fleet charge scheduling problem. The online heuristic makes fleet decisions for multiple regions, constrained resources, operational costs, renewable generation integration, and rebalancing.
The online heuristic does not rely on fractional allocations or rounding methods to produce integer assignments in a computationally feasible manner.
Additionally, online heuristic readily handles the inherent stochasticity of the fleet charge scheduling problem including unknown start locations, unknown energy levels, and unknown ride length/destination statistics.
Last, the online heuristic accounts for worst case scenarios (i.e., adversarially chosen between-ride sequences) and always yields welfare within a factor of $\frac{1}{\alpha}$ of the offline optimal. We discussed a competitive ratio (with proofs in the Appendix) and presented simulation results highlighting the performance of the heuristic.

%use section* for acknowledgment
% \section*{Acknowledgment}
% This work was supported by the California Energy Commission from an award administered through SLAC National Accelerator Laboratory. Solicitation: GFO-16-303. Agreement: EPC-16-057.

\bibliographystyle{IEEEtran}
\bibliography{references}
% \begin{IEEEbiography}[{\includegraphics[width=1in,height=1.24in,clip,keepaspectratio]{ahmad.jpg}}]{AHMADREZA  MORADIPARI}is pursuing the Ph.D. degree at the University of California Santa Barbara. He received the B.Sc. degree in Electrical Engineering from Sharif University of Technology in 2017. His research is mainly focused on optimization and control algorithms to manage congestion and reduce electricity costs in electric transportation systems.
% \end{IEEEbiography}
% \begin{IEEEbiography}[{\includegraphics[width=1in,height=1.21in,clip,keepaspectratio]{mahnoosh.png}}]{MAHNOOSH ALIZADEH}is an assistant professor of Electrical and Computer Engineering at the University of California Santa Barbara. Dr. Alizadeh received the B.Sc. degree in Electrical Engineering  from Sharif University of Technology in 2009 and the M.Sc. and Ph.D. degrees from the University of California Davis in 2013 and 2014 respectively, both in Electrical and Computer Engineering. From 2014 to 2016, she was a postdoctoral scholar at Stanford University. Her research interests are  focused on designing scalable control and data analytic frameworks and market heuristics for enabling sustainability and resiliency in societal infrastructure systems, with a particular focus on electric transportation systems.
% \end{IEEEbiography}
%\clearpage

\section*{Appendix}
\vspace{3ex}
\noindent\textbf{Definition 1. } (From \cite{IaaS}) \textit{The Differential Allocation-Payment Relationship for a given parameter $\alpha \geq 1$ is:}
\begin{align}
\label{eqn:diffalloc}
    \big(p(t) - f'(y(t))\big) \text{d}y(t) \geq \frac{1}{\alpha(t)} f^{*'}(p(t)) \text{d}p(t)
\end{align}
for all $t\in[0,T]$ and for all shared resources where $f'(y(t))$ is the derivative of an operational cost function and $f^{*'}(p(t))$ is the corresponding Fenchel conjugate's derivative.

\vspace{3ex}
\begin{lemma}
\label{lemma cables}
The online dual variable update heuristic \eqref{eqn:zero inf price} is $\alpha_1$-competitive in welfare generated from the limited number of SOMC EVSE cables where
\begin{align}
\nonumber
    \alpha_1= \ln{\Big(\frac{2\Psi U_c}{L_c}\Big)}.
\end{align}
\end{lemma}
\vspace{1ex}
\noindent \textit{Proof of Lemma \ref{lemma cables}}: We will show that the dual variable update heuristic in \eqref{eqn:zero inf price} satisfies the \textit{Differential Payment-Allocation Relationship} in equation \eqref{eqn:diffalloc} with parameter $\alpha_1$. Then the rest of the Lemma follows from Lemmas \ref{Background 1}, \ref{Background 2}, and \ref{Background 3}.

The cables at each SOMC EVSE do not have an explicit cost function; rather, they are free until the capacity $C_f$ is reached and then no more can be allocated. As such, the cost function $f(y_c^{mf}(t))$ for the cables can be seen as a zero-infinite step function with the step occurring right after $C_f$. Additionally, the dual variable update function \eqref{eqn:zero inf price} never allows $y_c^{mf}(t)$ to exceed $C_f$ so the derivative $f'(y_c^{mf}(t))=0$ while $y_c^{mf}\leq C_f$ and $y_c^{mf}(t) \leq C_f \;\forall\; m,f,t$.

The derivative of the Fenchel conjugate \eqref{eq:fenchel cable cost} for the cable resource is as follows:
\begin{align*}
%\label{eq:cable fenchel deriv}
    K_c^{mf*'}(p_c^{mf}(t)) = C_f.
\end{align*}
The derivative of the proposed pricing function \eqref{eqn:zero inf price} is:
\begin{align}
\label{eq:cable price deriv}
    \nonumber\text{d}p_c^{mf}(t) = &\Big(\frac{L_c}{2\Psi C_f}\Big)\Big(\frac{2\Psi U_c}{L_c}\Big)^{\frac{y_c^{mf}(t)}{C_f}}
    \ln{\Big(\frac{2\Psi U_c}{L_c}\Big)}\text{d}y_{c}^{mf}(t).
\end{align}
As such, after inserting $f'(y_c^{mf}(t))$, $ K_c^{mf*'}(p_c^{mf}(t))$, and $\text{d}p_c^{mf}(t)$ in \eqref{eqn:diffalloc}, we can show that the Differential Allocation-Payment Relationship holds when choosing $\alpha_1 = \alpha_c^{mf}(t)= \ln{\Big(\frac{2\Psi U_c}{L_c}\Big)}$. Because \eqref{eqn:diffalloc} holds for the update function, cost function, and Fenchel conjugate, the remainder of the proof follows from Lemmas \ref{Background 1}, \ref{Background 2}, and \ref{Background 3}.\hfill$\square$
\vspace{3ex}

\begin{lemma}
\label{lemma energy}
The online dual variable update heuristic \eqref{eqn:zero inf price energy} is $\alpha_2$-competitive in welfare generated from allocating the limited energy generation at each SOMC EVSE where:
\begin{align}
\nonumber
    \alpha_2= \ln{\Big(\frac{2\Psi U_e}{L_e}\Big)}.
\end{align}
\end{lemma}
\vspace{1ex}
\noindent \textit{Proof of Lemma \ref{lemma energy}}:
The proof is omitted, it is analogous to the proof of Lemma \ref{lemma cables}.

\vspace{1ex}
\begin{lemma}
\label{lemma generation}
The online dual variable update heuristic \eqref{eqn:newpricing} is $\alpha_3$-competitive in welfare when allocating and procuring electricity at each charging facility with procurement costs as in \eqref{eq:generation cost} where:
\begin{align}
\nonumber
    \alpha_3= \max_{\mathcal{F},\mathcal{T}}\Big\{\ln{\Big(\frac{2\Psi(U_g-\pi_f(t))}{L_g-\pi_f(t)}\Big)}\Big\}.
\end{align}
\end{lemma}
\vspace{3ex}
\noindent \textit{Proof of Lemma \ref{lemma generation}}:
We will show that the dual variable update heuristic in \eqref{eq:generation cost} satisfies the \textit{Differential Payment-Allocation Relationship} in equation \eqref{eqn:diffalloc} with parameter $\alpha_3$. Then the rest of the Lemma follows from Lemmas \ref{Background 1}, \ref{Background 2}, and \ref{Background 3}.

For the energy-procurement operational cost in \eqref{eq:generation cost} and its Fenchel conjugate \eqref{eq:fenchel gen cost} respectively, the following derivatives are:
\begin{align*}
    &G_f'(y_g^f(t)) = 
    \begin{cases}
        0, & y_g^f(t)\in [0, \delta_f(t)]\\
        \pi_f(t), & y_g^f(t)\in (\delta_f(t), \delta_f(t)+\mu_f(t)]
    \end{cases}\\
    &\nonumber\textrm{ and }\\
    &G_f^{*'}(p_g^f(t)) = 
    \begin{cases}
        \delta_f(t), & p_g^f(t)\in[0,\pi_f(t))\\
        \delta_f(t)+\mu_f(t), & p_g^f(t) \geq \pi_f(t).
    \end{cases}
\end{align*}
The derivative of the proposed pricing function \eqref{eqn:newpricing} is:
\begin{align*}
%\label{eqn:newpricing_deriv}
    &\text{d}p_g^f(y_g^f(t)) = \\*
    &\nonumber\begin{cases}
        \Big( \frac{L_g}{2\Psi\delta_f(t)} \Big) \Big( \frac{2\Psi\pi_f(t)}{L_g} \Big)^{ \frac{y_g^f(t)}{\delta_f(t)}}\\
        \nonumber\times\ln{\Big(\frac{2\Psi\pi_f(t)}{L_g}\Big)}\text{d}y_{g}^{f}(t), \hspace{69pt} y_g^f(t) < \delta_f(t), \\
        \Big( \frac{L_g-\pi_f(t)}{2\Psi(\delta_f(t)+\mu_f(t))} \Big) \Big( \frac{2\Psi(U_g-\pi_f(t))}{L_g-\pi_f(t)} \Big)^{ \frac{y_g^f(t)}{\delta_f(t)+\mu_f(t)}}\\
        \nonumber\times\ln{\Big(\frac{2\Psi(U_g-\pi_f(t))}{L_g-\pi_f(t)}\Big)}\text{d}y_{g}^{f}(t), \hspace{48pt}y_g^f(t) \geq \delta_f(t). \\
        \end{cases}
s\end{align*}

When $y_g^f(t) < \delta_f(t)$, $G_f'(y_g^l(t)) = 0$ and $G_f^{*'}(p_g^f(t)) = \delta_f(t)$. As such, after inserting the derivative $\text{d}p_g^f(y_g^f(t))$ in \eqref{eqn:diffalloc}, we can show that the Differential Allocation-Payment Relationship holds when $\hat{\alpha}_g^{f}(t)\geq \ln{\Big(\frac{2\Psi\pi_f(t)}{L_g}\Big)}$. 

Similarly, when $y_g^f(t) \geq \delta_f(t)$, $G_f'(y_g^f(t)) = \pi_f(t)$ and $G_f^{*'}(p_g^f(t)) = \delta_f(t)+\mu_f(t)$. As such, after inserting the derivative $\text{d}p_g^f(y_g^f(t))$ in \eqref{eqn:diffalloc}, we can show that the Differential Allocation-Payment Relationship holds when $\hat{\hat{\alpha}}_g^{f}(t)\geq \ln{\Big(\frac{2\Psi(U_g-\pi_f(t))}{L_g-\pi_f(t)}\Big)}$. Now, let $\alpha_3 = \alpha_g^{f}(t)=\max\{\hat{\alpha}_g^{f}(t), \hat{\hat{\alpha}}_g^{f}(t)\}$ and because \eqref{eqn:diffalloc} holds for the proposed pricing function, operational cost function, and Fenchel conjugate, the remainder of the proof follows from Lemmas \ref{Background 1}, \ref{Background 2}, and \ref{Background 3}.\hfill$\square$
\vspace{3ex}

\begin{lemma}
\label{lemma destination}
The online dual variable update heuristic \eqref{eqn:newpricing_reroute} is $\alpha_4$-competitive in welfare when routing EVs to each pick-up destination $d$ with AMoD limit $\Omega_d(t)$ where:
\begin{align}
\nonumber
    \alpha_4= \ln{\Big(\frac{2\Psi U_d}{L_d}\Big)}.
\end{align}
\end{lemma}
\vspace{3ex}
\noindent \textit{Proof of Lemma \ref{lemma destination}}: The proof is omitted, it is analogous to the proof of Lemma \ref{lemma cables}.
\vspace{1ex}

\begin{lemma}
\label{lemma outofservice}
The online dual variable update heuristic \eqref{eqn:newpricing_outofservice} is $\alpha_5$-competitive in welfare when assigning between-ride schedules to AMoD EVs and incurring out-of-service penalties as in \eqref{eq:outofservice cost} where:
\begin{align}
\nonumber
    \alpha_5= \max_{\mathcal{T}}\Big\{\ln{\Big(\frac{2\Psi(U_o-\phi(t))}{L_o-\phi(t)}\Big)}\Big\}.
\end{align}
\end{lemma}
\vspace{1ex}
\noindent \textit{Proof of Lemma \ref{lemma outofservice}}: The proof is omitted, it is similar to the proof of Lemma \ref{lemma generation} without needing to account for the piecewise operational cost.

\begin{lemma}
\label{Background 1}
(From \cite{IaaS}) If there is a constant $\alpha \geq 1$ such that the incremental increase of the primal and dual objective values differ by at most an $\alpha$ factor, i.e., $P^j - P^{j-1} \geq \frac{1}{\alpha}(D^j - D^{j-1})$, for every between-ride session $j$, then the heuristic is $\alpha$-competitive.
\end{lemma}
\noindent\textit{Proof of Lemma \ref{Background 1}}: Summing up the inequality at each step $j$, we have
\begin{align*}
    P^J &= \sum_j(P^j - P^{j-1})\geq \frac{1}{\alpha} \sum_j (D^j - D^{j-1})\\
    &=\frac{1}{\alpha}(D^J-D^0).
\end{align*}
Now, the initial primal value $P^0=0$ and by weak duality, $D^J \geq OPT$. Next, since $D^0=0$, we have that $P^J \geq \frac{1}{\alpha}OPT$. Thus, the online heuristic is $\alpha$-competitive. $\hfill\square$ 
%\newpage

\vspace{2ex}
\begin{lemma}
\label{Background 2}
(From \cite{IaaS}) If the \textit{Differential Allocation-Payment Relationship} holds for $\alpha\geq1$, then for each between-ride session $j$, the chosen schedule $s_j^{\star}$ guarantees that the increase in the fleet dispatcher's utility outweighs the operational cost by satisfying the following:
\end{lemma}
\begin{align*}
    \hat{p}_{js^{\star}} - \sum_{{t\in[t_j^-,t_{js^{\star}}^+]}} &\bigg( G_f(y_g^f(t))^{(j)} - G_f(y_g^f(t))^{(j-1)}\\
    &+O(y_o(t))^{(j)} - O(y_o(t))^{(j-1)}\bigg) \\
    &\geq \frac{1}{\alpha}(D^j - D^{j-1} - u_j)
\end{align*}
\noindent where 
\begin{align*}
    %\label{eq: payment}
    \hat{p}_{js^{\star}} = \;& p_d(t_{js^{\star}}^{+})d_{js^{\star}}^{+}(t_{js^{\star}}^{+}) \\
    &+ \sum_{{t\in[t_j^-,t_{js^{\star}}^+]}} \Big(o_{js^{\star}}(t)p_{o}(t) +  c_{js^{\star}}^{mf}(t)p_{c}^{mf}(t)\\
    &+e_{js^{\star}}^{mf}(t)(p_{e}^{mf}(t)+p_{g}^{f}(t)) \Big).
\end{align*}
\vspace{1ex}

\noindent\textit{Proof of Lemma \ref{Background 2}}: We expand out $D^j-D^{j-1} =$
\begin{align*}
    u_j + \sum_{t\in[t_j^-,t_{js^{\star}}^+]}  &\bigg( R_d^*(p_d(t))^{(j)} - R_d^*(p_d(t))^{(j-1)} \\
    &+ G_f^*(p_g^f(t))^{(j)} - G_f^*(p_g^f(t))^{(j-1)}\\
    &+ O^*(p_o(t))^{(j)} - O^*(p_o(t))^{(j-1)}\\
    &+ K_c^{mf*}(p_c^{mf}(t))^{(j)} - K_c^{mf*}(p_c^{mf}(t))^{(j-1)}\\
    &+ K_e^{mf*}(p_e^{mf}(t))^{(j)} - K_e^{mf*}(p_e^{mf}(t))^{(j-1)}.
\end{align*}
The lemma follows by summing the \textit{Differential Payment-Allocation Relationship} over all shared resources and over the entire time period.\hfill $\square$
\vspace{2ex}

\begin{lemma}
\label{Background 3}
(From \cite{IaaS}) If the Differential Allocation-Payment Relationship holds for $\alpha\geq1$ then $P^j-P^{j-1} \geq \frac{1}{\alpha} (D^j - D^{j-1}) $ for all $j$.
\end{lemma}
\vspace{2ex}

\noindent \textit{Proof of Lemma \ref{Background 3}}: If between-ride session $j$ results in the EV being sent back to the depot, then $P^j-P^{j-1} = D^j - D^{j-1}=0$. Otherwise, the change of the primal objective is $P^j - P^{j-1} = $
\begin{align*}
    v_{js^{\star}} - \sum_{t\in[t_j^-,t_{js^{\star}}^+]} &\Big( G_f(y_g^f(t))^{(j)} - G_f(y_g^f(t))^{(j-1)}\\
    &+O(y_o(t))^{(j)} - O(y_o(t))^{(j-1)}\Big)
\end{align*}
where $v_{js^{\star}} = u_j + \hat{p}_{js^{\star}}$. By Lemma \ref{Background 2}, we get that
\begin{align*}
    P^j - P^{j-1} \geq u_j + \frac{1}{\alpha}(D^j - D^{j-1} -u_j).
\end{align*}
With $u_j\geq 0$ and $\alpha\geq1$, then $P^j-P^{j} \geq \frac{1}{\alpha} (D^j - D^{j-1}) $.\hfill$\square$

\end{document}